\documentclass{article}
\usepackage[english]{babel}
\usepackage[letterpaper,top=2cm,bottom=2cm,left=3cm,right=3cm,marginparwidth=1.75cm]{geometry}

\usepackage{authblk}
\usepackage{algorithm,algorithmic}

\usepackage{amsmath,amscd,amsbsy,amssymb,latexsym,url,bm,amsthm}
\usepackage{dsfont}
\usepackage{color,xcolor,natbib}
\usepackage{graphicx}

\usepackage{tabularx}
\usepackage{times}
\usepackage{threeparttable}
\usepackage{multirow}
 \usepackage{booktabs}
\usepackage{amsmath,amscd,amsbsy,amssymb,latexsym,url,bm,natbib}
\allowdisplaybreaks
\usepackage{cleveref}
\usepackage{verbatim}
\usepackage{color}
\usepackage{dsfont}
\usepackage{caption}
\usepackage{algorithm}
\usepackage{algorithmic}

\newtheorem{definition}{Definition}
\newtheorem{lemma}{Lemma}
\newtheorem{theorem}{Theorem}

\usepackage{thmtools}
\usepackage{thm-restate}
\usepackage{mathtools}

\newcommand{\EE}[1]{\mathbb{E} \left[#1\right]}

\newcommand{\cF}{\mathcal{F}}

\newcommand{\cS}{\mathcal{S}}

\newcommand{\abs}[1]{\left| #1 \right|}
\newcommand{\bOne}[1]{\mathds{1} \! \left\{#1\right\}}
\newcommand{\bracket}[1]{\left(#1\right)}

\newcommand{\set}[1]{\left\{ #1 \right\}}
\newcommand{\PP}[1]{\mathbb{P} \left(#1\right)}

\newcommand{\etal}{\emph{et al.}}

\newcommand{\opt}{\mathrm{opt}}

\newif\ifsup\supfalse
\suptrue

\newcommand{\oracle}{\mathrm{Oracle}}

\DeclareMathOperator*{\argmax}{argmax}

\begin{document}

\title{Online Influence Maximization under Decreasing Cascade Model}

\author[1]{Fang Kong}
\author[1]{Jize Xie}
\author[2]{Baoxiang Wang}
\author[3]{Tao Yao\thanks{Work done at Alibaba Group, and now affiliated with the Chinese University of Hong Kong, Shenzhen.}}
\author[1]{Shuai Li \thanks{Corresponding author.}}

\affil[1]{Shanghai Jiao Tong University}
\affil[2]{The Chinese University of Hong Kong, Shenzhen}
\affil[3]{Alibaba, Inc.}

\affil[ ]{\{fangkong,xjzzjl,shuaili8\}@sjtu.edu.cn, bxiangwang@cuhk.edu.cn, tao.yao@alibaba-inc.com}
\date{}

 \maketitle

\begin{abstract}%
We study online influence maximization (OIM) under a new model of decreasing cascade (DC). This model is a generalization of the independent cascade (IC) model by considering the common phenomenon of market saturation. In DC, the chance of an influence attempt being successful reduces with previous failures. The effect is neglected by previous OIM works under IC and linear threshold models. 
We propose the DC-UCB algorithm to solve this problem, which achieves a regret bound of the same order as the state-of-the-art works on the IC model. 
Extensive experiments on both synthetic and real datasets show the effectiveness of our algorithm.
\end{abstract}

\section{Introduction}\label{sec:intro}

The study of information diffusion on social networks has received increasing attention from the community of machine learning, data mining, and graph algorithms.
A term to be diffused has many forms, including the spread of news and opinions, adoption of products, and broadcast of alarms.
To characterize this process, many influence propagation models have been proposed. 
Among them, the independent cascade (IC) model and the linear threshold (LT) model are widely adopted \citep{chen2013information,Kempe2003,wortman2008viral,gruhl2004information,chen2010scalableLT}. 
The common ground of all these models is to use a weighted graph to represent a social network, where the weights denote influence abilities between nodes. 
For example, the IC model assumes that the term transmits through each edge independently with a probability equal to the weight; under LT, the transmission happens when the cumulative weight of in-neighbors exceeds a certain threshold.
Despite the popularity, both IC and LT neglect the decay characterization, which is common in many real applications 
\citep{kempe2005influential,wortman2008viral,steeg2011stops,zhang2016influence}. 
This decay property reflects the phenomenon of {\it market saturation} where more failed influence attempts would turn the market to be more saturated, making subsequent influence trials less likely to succeed.
The decreasing cascade (DC) model \citep{kempe2005influential} is a generalization of IC that depicts these observations for better feasibility.

When the term is the alarm of an emergency or a broadcast of an important notification, the relevant party has a crucial responsibility to spread it to the possible extent under limited resources. Similarly, the marketing team of a company advertises products on social networks, aiming to attract as many users as possible \citep{wortman2008viral,Kempe2003}. 
% datta2010viral,
These real tasks motivate the problem of influence maximization (IM) \citep{Kempe2003,IMM2015,wang2012scalableIC,chen2010scalableLT}.
Given a graph and its underlying influence propagation model together with the model parameters (the graph weights), it desires to find an initial set of adopted users (the seed set) to maximize the influence spread. The IM problem has been widely studied under IC \citep{wang2012scalableIC,jung2012irie}, LT \citep{chen2010scalableLT,goyal2011simpathLT} and the DC model \citep{kempe2005influential}.

A major concern of IM is that in real applications, the parameters of the influence propagation models are usually unknown. For example in advertisement placing, a company might not know the actual influence probabilities before they place the advertisement. 
A heuristic to address this issue is to estimate the unknown parameters from the collected past observations \citep{netrapalli2012learning,goyal2010learning}. However, there might not exist sufficient logs, and even if they exist, the logs have biases. 
Also, such estimates cannot adapt to any change in the social network. 
Online influence maximization (OIM), instead, learns the unknown parameters through the iterative interactions with the social network and progressively finds the optimal seed set \citep{lei2015online,WeiChen2016,WeiChen2017,zhengwen2017nips,IMFB2019,Model-Independent2017,li2020online}. 
Previous theoretical OIM studies mainly focus on the IC and LT models \citep{WeiChen2013,WeiChen2017,WeiChen2016,zhengwen2017nips,IMFB2019,Sharan2015-nodelevel,li2020online}. Despite their importance, these model assumptions fail to characterize the common damping phenomenon of market saturation in influence spreading \citep{zhang2016influence,wortman2008viral,kempe2005influential,steeg2011stops}. Thus, how to design efficient algorithms on the more general DC setting remains an open problem.

We are the first to formulate the OIM problem under the DC model.
In this paper, we propose an upper confidence bound (UCB)-based algorithm, DC-UCB, to solve this problem. 
The algorithm meticulously readjusts the UCB indices of activation probabilities to maintain the decreasing property of DC.
Through careful analysis of the information diffusion process under DC, we prove a DC-based triggering probability modulated (TPM) bounded smoothness condition, as an analogy to that under IC \citep{WeiChen2017,zhengwen2017nips}. 
We can then provide rigorous theoretical guarantees on the regret of DC-UCB. 
The regret upper bound of DC-UCB achieves the same order as the state-of-the-art results under IC. 
Extensive experiments on both synthetic and real-world networks show the effectiveness and efficiency of our algorithm.

\section{Related Work}

The (offline) IM problem was formulated by \cite{Kempe2003}.
It conducts discrete optimization on the seed set to maximize the influence spread on graphs when given parameters of the underlying diffusion model. Since then, this problem has attracted a lot of attention \citep{kempe2005influential,chen2010scalableLT,wang2012scalableIC,IMM2015,zhang2016influence} focusing on different underlying diffusion models and different solving techniques.
The adaptive influence maximizationn (AIM) problem is a variant of IM where the agent can adaptively select seed nodes after it observes the propagation results of previously selected seeds \citep{han2018efficient,cautis2019adaptive}. Both problems assume that the diffusion parameters are known beforehand.

When the parameters of the diffusion model are unknown, the problem can be solved through online IM (OIM), aiming to learn the parameters through the interactions with the social network. 
The framework of OIM can be formulated as a problem of combinatorial multi-armed bandits (CMAB) \citep{WeiChen2013,WeiChen2016,WeiChen2017} - a $T$-round game between the learning agent and the environment to maximize the cumulative reward. 
In each round, the learning agent executes a combination of base arms, defined as a super arm, and observes the feedback (i.e. the influence propagation in OIM).
Based on the collected feedback, it then updates its knowledge for the unknown environmental parameters and improves the subsequent choices.

Chen $\etal$ are the first to use the CMAB framework with probabilistically triggered arms to study the OIM problem under the IC model with edge-level feedback 
\citep{WeiChen2013,WeiChen2016,WeiChen2017}. In this formulation, each edge is regarded as a base arm and all outgoing edges from the seed set are regarded as a super arm. 
The line of studies proposes a CUCB algorithm based on the canonical upper confidence bound (UCB) algorithm \citep{auer2002finite} and derives rigorous guarantees of it \citep{WeiChen2017}.
To generalize it to large-scale real applications, \cite{zhengwen2017nips} present a linear variant for the activation probabilities and propose the IMLinUCB algorithm.  
\cite{IMFB2019} consider the network assortativity and assume the activation probability of each edge can be decomposed by the influence factor of the source node and the susceptibility factor of the destination node to reduce the sample complexity. 
Node-level feedback, which needs less information and is more realistic than the edge-level feedback, has also been considered under the IC model \citep{Sharan2015-nodelevel}. It assumes that the identities of influenced nodes (instead of transmitted edges) can be observed, and provides a bound on 
estimation gap of the activation probabilities between node-level feedback and edge-level feedback.
Recently, \cite{zhang2022online} give a regret upper bound for this challenging feedback.

A few OIM works consider different diffusion models. 
Until recently, the OIM problem under the LT model was solved by \cite{li2020online}. This work assumes the full node-level feedback, the influence status of each node in each diffusion step, can be observed and gives the first regret upper bound under this model.
Another work \citep{Model-Independent2017} considers a pairwise feedback scheme, where the agent can directly observe the influence status between each node and each seed node. Though the setting can be applied to many diffusion models like IC, LT, and DC, there are no optimality guarantees for their heuristically proposed objective function. Our work is the first one to study the OIM problem under the DC model with rigorous theoretical guarantees.

\section{Setting}\label{sec:setting}

In this section, we formulate the OIM problem under the DC model. 
The social network is the basic structure of this problem, which is usually represented by a directed graph $G=(V,E)$ with the node set $V$ and the edge set $E$ denoting the set of users and the set of relationships between users, respectively. An edge $e = (u,v) \in E$, for example in Twitter, can correspond to the relationship of user $v$ following user $u$ and the information transmitting from $u$ to $v$. Let $n = \abs{V}$ and $m=\abs{E}$ be the number of nodes and edges, respectively. For each node $v \in V$, denote $N(v) := N^{\mathrm{in}}(v)$ as the set of all incoming neighbors of $v$, abbreviated as in-neighbors. 

The IC and LT models \citep{Kempe2003} are two of the most common and widely used influence propagation models in characterizing the information diffusion on social networks \citep{wang2012scalableIC,chen2010scalableLT,li2020online}. 
% ThresholdModel1978,ThresholdModel2007,,jung2012irie
Under the IC model, each node will try to activate all its inactive out-neighbors independently right after it is activated. 
The success probability of activation attempt between every such pair is equivalent to the weight of this edge. 
Under LT, a node is activated if the sum of edge weights from its active in-neighbors exceeds a certain threshold. 

However, the information diffusion can be very complicated in real applications. One of the main considerations is the famous effect of {\it market saturation} in real life \citep{kempe2005influential,wortman2008viral,steeg2011stops,zhang2016influence,leskovec2007dynamics,kossinets2006empirical}. As the information spreads more around the neighborhood, users usually become more saturated with the marketing and their in-neighbors will have diminishing influence effects on them. For example in the diffusion of a news story, the story would become more redundant and unattractive every time the user sees it from the broadcasts of the neighbors but expresses no interest, making the adoption probability decrease with the time of exposures \citep{wu2007novelty,hogg2009stochastic,myers2012information}. 
Such characterization of influence diffusion is known as the decay property \citep{zhang2016influence,wortman2008viral,steeg2011stops}, which is not fully covered in the common IC and LT models but can be captured in the DC model \citep{kempe2005influential}.

The information diffusion process of the DC model starting from the seed set $S$ is described as follows. Define $S_{\tau}$ as the set of influenced nodes until the end of time step $\tau$. In the beginning when $\tau=0$, only nodes in $S$ are influenced, that is, $S_0 = S$. 
Then after time step $\tau \ge 0$ for each inactive node $v \in V\setminus S_{\tau}$, all of its active in-neighbors who are influenced at the last time step, i.e. elements in $N(v) \cap \bracket{S_{\tau}\setminus S_{\tau-1} }$, will make an attempt to activate $v$ in an arbitrary order (denote $S_{-1} = \emptyset$ for consistency). 
Specifically, each node $u \in N(v) \cap \bracket{S_{\tau} \setminus S_{\tau-1} }$ tries to activate $v$ with probability $p_v(u, S')$, where $S'$ is the set of nodes that have already tried but failed to activate $v$ in all previous steps ($<\tau$) and the current step ($=\tau$).
If there exists an $u$ who successfully activates $v$ at $\tau+1$, then $v$ becomes active, or equivalently $v \in S_{\tau+1}\setminus S_{\tau}$; otherwise if all the nodes in $N(v) \cap \bracket{S_{\tau}\setminus S_{\tau-1} }$ fail to activate $v$, then $v$ is still inactive at $\tau+1$, or equivalently $v \notin S_{\tau+1}$. 
Such an information diffusion process ends when no node is activated at a new step. The {\it influence spread} $r(S,p)$ is defined as the expected number of total influenced nodes under the seed set $S$ and the activation probabilities $p = (p_v(u, S'))_{v \in V, u \in N(v), S' \subseteq N(v) \setminus \set{u}}$.
Here the expectation is taken over the randomness in the diffusion process, specifically the success or failure events of all activation attempts.

The activation probabilities $p$ under the DC model satisfy the following two mild but important properties.\\
\indent \textbf{Decreasing. } The activation probability of $u \in N(v)$ on $v$ decreases with more previous failed attempts. Specifically, if $S' \subseteq S'' \subseteq N(v) \setminus \set{u}$, then $p_v(u,S')\ge p_v(u,S'')$.\\
\indent \textbf{Order-independence. } The probability that $v$ is eventually influenced by the set $S' = \set{u_1,u_2,\ldots,u_\ell} \subset N(v)$ does not depend on the order of these nodes' activation attempts. That is, the probabilities that $S'$ successfully activate $v$ in order of $u_1,u_2,\ldots,u_\ell$ and $u_\ell,u_{\ell-1},\ldots,u_1$ are the same.

With seed set cardinality $K$, denote $\cS = \set{S \subset V: \abs{S} \le K}$ as the action set which consists of all feasible seed sets with size smaller than $K$. When the activation probability vector $p$ is known, the (offline) IM problem aims to find an $S \in \cS$ with the maximum influence spread $\argmax_{S\in \cS} r(S,p)$. This problem under the DC model is NP-hard but can be approximately solved with a greedy algorithm, since the influence spread function is monotone and submodular \citep{kempe2005influential}. 
We use $\opt_p = \max_{S \in \cS}r(S,p)$ and $S_p^{\opt} \in \argmax_{S \in \cS} r(S,p)$ to denote the maximum influence spread and an optimal seed set, respectively. Let $\oracle$ be an offline IM algorithm that outputs a seed set given the activation probabilities $p$. For $\alpha, \beta \in [0,1]$, we say $\oracle$ is an $(\alpha,\beta)$-approximation if its output $S^\ast = \oracle(p)$ satisfies $\PP{r(S^\ast,p)\ge \alpha \cdot \opt_p} \ge \beta$ for any input $p$. 

It is worth pointing out that the IC model is a special case of the DC model since it satisfies that $p_v(u,S') = p_{u,v}$ for any edge $(u,v) \in E$ and the two properties of the DC model can be verified easily. In this paper, we consider another special case of DC where $p_v(u,S')$ only depends on the size of $S'$ and node $v$, not on $u$ and the elements in $S'$. 
It characterizes that the probability of a node being influenced depends on the susceptibility of the node itself \citep{watts2007influentials} and the previous failed attempts. This setting keeps the most important decay property to describe the real-world phenomena of market saturation.
% And it is easy to verify that the two properties of decreasing and order-independence are satisfied in this setting.
Specifically, each node $v$ is associated with a decreasing probability sequence of size $\abs{N(v)}$, denoted as $p_v := [p_{v}(1),p_{v}(2), \ldots, p_{v}(\abs{N(v)})]$. Then the probability that $u$ successfully activates $v$ after the attempts of the nodes in $S'$ is $p_v(u, S') = p_{v}(\abs{S'}+1)$.
The activation probability vector can thus be written as $p := (p_v(i))_{v \in V, i\in [\abs{N(v)}]}$ and it is immediate to verify that two properties of the DC model hold. 
Note that if for each node $v$, its related activation probabilities are the same, or $p_v(i)\equiv p_v'$ for all $i \in [|N(v)|]$, the diffusion process under this specific DC model is the same with that under IC.

In the online version where the activation probability vector $p$ is unknown, the problem aims to learn those unknown probabilities from the interactions with the social network and to gradually identify the optimal seed set. In each round $t$, the learning agent selects a seed set $S_t \in \cS$. Then the diffusion process originating from $S_t$ could reveal some influence propagation, based on which the agent could get some information about the unknown parameters.
Similar to most OIM works, we consider the (partial) edge-level feedback where an edge is observed only when its start node is active and the end node is inactive. Recall that the (full) edge-level feedback assumes an edge to be observed if its start node is active \citep{WeiChen2017,WeiChen2016,zhengwen2017nips,IMFB2019}. Our (partial) edge-level feedback requires less information since each node is activated at most once. It is more reasonable that the following activation attempts on a node after it is activated are not supposed to be observed.

With an $(\alpha,\beta)$-approximation oracle, the objective of the learning agent is to maximize the $T$-round cumulative influence spread, or equivalently to minimize the cumulative $\alpha\beta$-scaled regret \citep{WeiChen2016,WeiChen2017,zhengwen2017nips,IMFB2019} over $T$ rounds
\begin{align}
	R(T) = \mathbb{E}\bigg[\sum_{t=1}^T r(t)\bigg] =  \mathbb{E}\bigg[\sum_{t=1}^T \bracket{\alpha \beta \cdot \opt_p - r(S_t,p)} \bigg]\,,
\end{align} 
where $r(t)$ is the regret at $t$ and the expectation is taken over the randomness in diffusion processes and the adopted oracle.

\section{The DC-UCB Algorithm}\label{sec:DCUCB}

In this section, we introduce DC-UCB (Algorithm \ref{alg:DCUCB}), a UCB-type algorithm, to solve the OIM problem under the DC model with (partial) edge-level feedback.

The DC-UCB algorithm takes the graph $G=(V,E)$, the seed set cardinality $K$ as well as an offline $\oracle$ as input. For each $p_v(i)$ that represents the success probability of the $i$-th activation attempt on node $v$, the algorithm maintains its empirical mean $\hat{p}_v(i)$ and the number of observations $T_v(i)$.

\begin{algorithm}[thb!]
\caption{DC-UCB}\label{alg:DCUCB}
\begin{algorithmic}[1]
\STATE \textbf{Input:} Graph $G=(V,E)$; seed set size $K$; $\oracle$
\STATE \textbf{Initialize}: $\hat{p}_v(i)=0, T_{v}(i) = 0$, for $v \in V, i \in [\abs{N(v)}]$
\FOR {$t = 1,2, \ldots$}
	\FOR {each node $v \in V$, $i =  1,2,3,\ldots,\abs{N(v)}$} \label{alg:newUCB:start}
		% \FOR {}
%			\STATE $\rho_v(i) = \sqrt{\frac{3 \log t}{2 T_v(i)}}\ $ (or $\rho_v(i) = \infty$ if $T_v(i)=0$) \shuai{what initialization?}
			\STATE $\bar{p}_v(i) = \text{Proj}_{[0,1]}\bracket{ \hat{p}_v(i) +  \sqrt{\frac{3 \log t}{2 T_v(i)}} }  $ \label{alg:initialUCB}
			\STATE $\bar{p}'_v(i) = \min \set{\bar{p}'_v(i-1), \bar{p}_v(i)}$ \quad ($\bar{p}'_v(0)=1$) \label{alg:scaleUCB}
		% \ENDFOR
%		\STATE $\bar{p}'_v(1) =  \bar{p}_v(1)$ \label{alg:scaleUCB:start}
%		\FOR {$i =  2,3,\ldots,|N(v)|$}
%			\STATE $\bar{p}'_v(i) = \min \set{\bar{p}'_v(i-1), \bar{p}_v(i)  }$
%		\ENDFOR \label{alg:scaleUCB:end}
	\ENDFOR \label{alg:newUCB:end}
	\STATE Choose $S_t = \oracle(G,K,\bar{p}')$ and observe feedback \label{alg:UCB:choose}
	\FOR {each status $Y_v(i)$ of all attempts}
		\STATE $\hat{p}_v(i) = \frac{T_v(i)\cdot \hat{p}_v(i) + Y_v(i)}{T_v(i)+1}$; $T_v(i) = T_v(i)+1$ \label{alg:UCB:update}
		 % \bracket{T_v(i)\cdot \hat{p}_v(i) + X_v(i) }/\bracket{T_v(i)+1} $
	\ENDFOR
\ENDFOR
\end{algorithmic}
\end{algorithm}

In each round $t$, the learning agent first computes the UCB $\bar{p}_v(i)$ for each activation probability $p_v(i)$ based on collected observations (line \ref{alg:initialUCB}). The computation of the UCBs are based on the Chernorff-Hoeffding inequality \citep{hoeffding1963probability} to guarantee the value is an upper bound of the true value with high probability, which is applicable here since observations on the same term $p_v(i)$ in different rounds are independent. 
The operating $\text{Proj}_{[0,1]}(\cdot)$ projects a real number into interval $[0,1]$ to ensure the UCBs of probabilities still fall into this interval. Specially, if $T_v(i)=0$ for $i$-th activation probability of node $v$, we simply set $\bar{p}_v(i)=1$.
Since the decreasing property of the DC model guarantees that the $i$-th real probability always larger than the $i+1$-th real probability for any node $v$, 
these UCB indices are then capped to maintain the decreasing property (line \ref{alg:scaleUCB}). 

With the capped UCBs $\bar{p}'$, graph $G$ and seed set size $K$ as input, the offline $\oracle$ computes a seed set $S_t$ (line \ref{alg:UCB:choose}). 
The returned solution automatically balance the exploitation and exploration: If all activation probabilities are observed enough, then their UCBs are roughly the empirical means and the $\oracle$ will return a solution that is approximately optimal under the estimated weights, whose value is close to the one under true probabilities since the influence spread is continuous in vector $p$;
% similarly with it under the real weights due to the close distance from the estimated weights to real weights;
% the continuity of influence spread on the input weight vector and 
if some activation probabilities are not observed enough, their confidence interval would be wide and their UCBs would be high, making the graph lean towards these less-explored parts and thus forcing the exploration. 

Then the influence spreads from the selected seed set $S_t$. And the agent can observe a binary variable $Y_v(i)$ if there is the $i$-th attempt to activate $v$,
% as the result of the $i$-th attempt to activate $v$, 
where $1$ represents the successful activation and $0$ represents failure.
With the $Y_v(i)$, the corresponding $\hat p_v(i)$ is updated (line \ref{alg:UCB:update}).

\subsection{Regret Bounds}

The following theorem shows the problem-independent regret bound for our algorithm DC-UCB.

\begin{theorem}\label{thm:UCB:free}
	The $(\alpha,\beta)$-scaled regret of DC-UCB satisfies
	\begin{align}
    R(T) &\le O\bracket{nm\sqrt{T\log T}}\,.
	\end{align}
\end{theorem}

This is the first theoretical result of the OIM problem under DC, which does not require the strong independence assumption in the IC model and considers the common market-saturation phenomenon in real life. 
% Since the diffusion process under the DC model is the same as it under the IC model when all activation probabilities are the same, our result is also comparable to the OIM works on the IC model.
Compared with the regret bound of the IC model, our regret achieves the same order in the graph parameters $n,m$ and time horizon $T$ \citep{WeiChen2017}, though strictly speaking the regret bounds under two different models are not directly comparable. Recall that if the activation probabilities satisfy $p_v(i) \equiv p_v'$, the diffusion process would be the same under the IC model and the regret bounds under this case can be directly comparable.

Besides problem-independent bound (Theorem \ref{thm:UCB:free}), we also provide the problem-dependent bound for DC-UCB. 

To get this, define the gap between the influence spread of a seed set $S$ and the $\alpha$-scaled optimal influence spread as
\begin{align*}
	\Delta_S = \max\set{0, \quad \alpha \cdot \opt_p - r(S, p)}\,,
\end{align*}
where $\alpha$ is the approximation factor of the offline $\oracle$.
And for each entry $(v, i)$ with $v \in V,i \in [\abs{N(v)}]$, define $P_{v, i}^{S}$ as the probability that node $v$'s $i$-th attempt can be observed in the diffusion process starting from $S$ under activation probability vector $p$ (here we omit the dependence on $p$ for simplicity). 
Then we can define the arm gap with the aid of such observation probability
\begin{align*}
	\Delta_{\min}^{v,i} = \inf_{S\in \cS:\ \Delta_S>0, \ P_{v,i}^{S}>0} \ \Delta_S
	% \,, \quad \Delta_{\max}^{v,i} = \sup_{S\in \cS:\ \Delta_S>0, \ P_{v,i}^{S}>0} \ \Delta_S
\end{align*}
and take the minimum over all entries
\begin{align*}
	\Delta_{\min} = \min_{v \in V,\ i\in [|N(v)|]} \  \Delta_{\min}^{v,i} \,.
	% \,,\quad  \Delta_{\max} = \max_{v \in V,\ i\in [|N(v)|]}  \ \Delta_{\max}^{v,i}\,.
\end{align*}

Similarly to the IC model, 
let 
\begin{align}
\widetilde{V} = \max_{u \in V} \sum_{v \in V}\bOne{\text{there is a path from } u \text{ to } v}
\end{align}
be the maximum number of nodes that a node can reach in $G$. 
With these notations, the problem-dependent regret bound is provided in the next theorem.
\begin{theorem}\label{thm:UCB:dependent}
	The $(\alpha,\beta)$-scaled regret of the DC-UCB algorithm can be bounded as
	\begin{align}
		R(T) \le 
		& \ \sum_{v \in V,\ i \in [|N(v)|]} \frac{576{\widetilde{V}}^2m\ln T}{\Delta_{\min}^{v,i}}   \notag \\
		& \quad + \frac{\pi^2 \Delta_{\max}}{6}\sum_{v \in V,\ i \in [|N(v)|]}\bracket{2+\log \frac{4m\widetilde{V}} {\Delta_{\min}^{v,i}}  }+4m\widetilde{V} \\
		& \ O\bracket{\frac{m^2n^2}{\Delta_{\min}} \ln T}\,.
	\end{align}
\end{theorem}

Due to the space limit, the detailed proof of Theorem \ref{thm:UCB:free} and \ref{thm:UCB:dependent} are provided in Appendix. 
% \footnote{The full version is available at \url{https:xx}.}. 
By carefully analyzing the information diffusion process under DC, we prove our DC-based TPM condition, similar to that under IC \citep{WeiChen2017,zhengwen2017nips}. Such TPM condition bounds the difference between influence spread under different activation probabilities, which is crucial to acquire the above theoretical guarantees.

\subsection{The TPM Condition under the DC model}

To bound the influence spread difference $r(S_t,\bar{p}_t) - r(S_t,p)$ under two activation probabilities, we need the following key theorem of the triggering probability modulated (TPM) condition. Such a condition is crucial in deriving the final regret bound, similarly to that under the IC model \citep{WeiChen2017,zhengwen2017nips}.
Denote $V_{S,v}$ as the set of vertices who are on any path from $S$ to $v$ for any seed set $S$.
\begin{theorem}
\label{lem:TPM}
For any two activation probability vectors $p$ and $\bar{p}$ satisfying $p_v(i) \le \bar{p}_v(i)$ for any $v \in V, i \in [|N(v)|]$, the difference between the influence spread of any seed set $S$ under these two activation probability vectors is at most
\begin{align}
	 & \  r(S,\bar{p}) - r(S,p)\nonumber \\
	\le\  &\ \mathbb{E}\bigg[\sum_{v \in V\setminus S} \ \sum_{u \in V_{S,v}} \ \sum_{i \in [|N(u)|]} \ \bOne{O_u(i)} \cdot[\bar{p}_u(i) - p_u(i)]\bigg] \label{thm:keylm-lin} \\
	=\  & \  \sum_{v \in V\setminus S} \ \sum_{u \in V_{S,v}} \ \sum_{i \in [|N(u)|]} \ P_{u,i}^S \cdot[\bar{p}_u(i) - p_u(i)] \label{thm:keylm}\,,
\end{align}
where $O_u(i)$ denotes the event that the $i$-th attempt to activate $u$ under $p$ can be observed.
\end{theorem}

\begin{proof}[Proof of Theorem \ref{lem:TPM}]
	Recall that $r(S,p,v)$ is the probability that $v$ is finally influenced in the diffusion process starting from $S$ under $p$. We can decompose the influence spread difference under two activation probability vectors as
	\begin{align*}
		r(S,\bar{p}) - r(S,p) =& \sum_{v \in V\setminus S} r(S,\bar{p},v) - r(S,p,v)  \\
		=& \mathbb{E}\left[ \sum_{v \in V\setminus S} \bOne{v \text{ is influenced under }\bar{p}} \right.\\
  &\left. - \bOne{v \text{ is influenced under }p}\right] \,.
	\end{align*}

According to the monotonicity of the influence spread, \\
$$\bOne{v \text{ is influenced under }\bar{p}} - \bOne{v \text{ is influenced under }p} \in \set{0,1}.$$ 
When $\mathds{1} \{v$ $\text{ is influenced under }\bar{p}\}$ $- \bOne{v \text{ is influenced under }p} = 1$, which means 
\begin{align*}
    &\bOne{v \text{ is influenced under }\bar{p}} = 1\,,\\
     &\bOne{v \text{ is influenced under }p} = 0\,,
\end{align*}
$v$ is influenced under $\bar{p}$ but not influenced under $p$. 

Since the influence status of $v$ under two activation probabilities are different, there must exist a step $\tau>0$ in the diffusion process such that starting from $\tau$, the influence statuses of the nodes related to $v$ are different. To be specific, there must exist some time step $\tau$ and node $u \in V_{S,v}$ such that the active in-neighbors of $u$ under $p$ and $\bar{p}$ are the same until the end of $\tau$, but the influence status of $u$ under $p$ and $\bar{p}$ are not the same at $\tau+1$. 

Let $\ell_{u,\tau-1}, \ell_{u,\tau}$ be the number of active in-neighbors at step $\tau-1$ and $\tau$, respectively. Then at step $\tau$, there are totally $(\ell_{u,\tau}-\ell_{u,\tau-1}-1)$ attempts to activate $u$ under both $p$ and $\bar{p}$. We denote them as the $(\ell_{u,\tau-1}+1)$-th, $(\ell_{u,\tau}+2)$-th, $\ldots$, $\ell_{u,\tau}$-th attempt. Also for any $\ell_{u,\tau-1}+1 \le j \le \ell_{u,\tau}$, let $Y_u(j)$ and $\bar{Y}_u(j)$ be the status of the $j$-th attempt under $p$ and $\bar{p}$, respectively. It holds that $\EE{Y_u(j)} = p_u(j), \EE{\bar{Y}_u(j)} = \bar{p}_u(j)$. 

Then based on the above analysis, there must exists an $i$-th attempt, where $\ell_{u,\tau-1}+1 \le i \le \ell_{u,\tau}$ and all attempts to activate $u$ before $i$ fail under both $p$ and $\bar{p}$ and the $i$-th attempt succeeds under $\bar{p}$ but fails under $p$. That is,
	\begin{align*}
		\forall j <i, X_u(j) = 0, \bar{X}_u(j) = 0 \,,\\
		X_u(i) = 0, \bar{X}_u(i) = 1 \,.
	\end{align*}
This event happens with probability $\prod_{j < i} (1-\bar{p}_u(j)) \cdot [\bar{p}_u(i) - p_u(i)] $. By the union bound, we conclude that
	\begin{align}
		&r(S,\bar{p}) - r(S,p) \\
  =& \mathbb{E}\left[ \sum_{v \in V\setminus S} \bOne{v \text{ is influenced under }\bar{p}} \right.\\
  &\left. - \bOne{v \text{ is influenced under }p} \right] \notag \\
\le& \EE{ \sum_{v \in V\setminus S} \sum_{u \in V_{S,v}} \sum_{\ell_{u,\tau-1}+1 \le i \le \ell_{u,\tau}  }\prod_{j < i} (1-\bar{p}_u(j)) \cdot [\bar{p}_u(i) - p_u(i)] } \notag \\
\le& \EE{ \sum_{v \in V\setminus S} \sum_{u \in V_{S,v}} \sum_{i \in [|N(u)|]} \bOne{A_u(i)}\prod_{j < i} (1-\bar{p}_u(j)) \cdot [\bar{p}_u(i) - p_u(i)] } \label{eq:keylemma:defA} \,,
	\end{align} 
where \eqref{eq:keylemma:defA} is derived by the definition of event $A_u(i)$, which denotes that there are more than $i$ active in-neighbors of node $u$ under $p$. 

Also based on the feedback scheme of both the DC-UCB algorithm and the DC-LinUCB algorithm, all status of attempts to activate $u$ before $u$ is finally influenced can be observed, which means that 
\begin{align*}
	\EE{\bOne{A_u(i)}\prod_{j < i} (1-\bar{p}_u(j))} &\le \EE{ \bOne{A_u(i)}\prod_{j < i} (1-{p}_u(j))} \\
 &= \EE{\bOne{O_u(i)}} = P_{u,i}^S \,.
\end{align*}
Thus we conclude that 
\begin{align*}
	&r(S,\bar{p}) - r(S,p)\\ 
 \le& \EE{ \sum_{v \in V\setminus S} \sum_{u \in V_{S,v}} \sum_{i \in [|N(u)|]} \bOne{O_u(i)} \cdot [\bar{p}_u(i) - p_u(i)] } \\
=& \sum_{v \in V\setminus S} \sum_{u \in V_{S,v}} \sum_{i \in [|N(u)|]} P_{u,i}^S \cdot[\bar{p}_u(i) - p_u(i)] \,.\qedhere
\end{align*} 
\end{proof}

\section{Experiments}\label{sec:exp}

In this section, we compare our DC-UCB algorithm with related baselines in both synthetic and real-world networks\footnote{The code is available at \url{https://github.com/fangkongx/OIM-DC}.}. Since this work is the first to study the OIM problem under the DC model with rigorous theoretical guarantees, few directly comparable baselines exist. We exhaust those baselines and also adopt baseline methods that give insight into the performance of our algorithms via indirect comparisons. The following are descriptions. 

% \subsection{Baseline Algorithms}

 % \textit{UCB.~} 
\indent \textit{UCB.~} The OIM problem under DC can be regarded as a multi-armed bandit problem if we treat each seed set as an arm and the influence spread as its expected reward. 
Thus, the classical UCB algorithm \citep{auer2002finite} can be applied to solve this problem. This algorithm maintains a UCB index for each arm and selects the arm enjoying the highest UCB index in each round. It is acceptable when the graph is small but is not feasible when the graph is large as the number of arms grows exponentially with the number of nodes.

 % \textit{CMAB-UCB-average and CMAB-UCB-random. ~}
\indent \textit{CMAB-UCB-average and CMAB-UCB-random.~} The OIM problem under DC is a CMAB problem, where we can treat each node as a base arm and a seed set as a super arm. The influence spread of a seed set is then the expected reward of the super arm. Note that this reward cannot be written in a linear function form in terms of the utilities of the included base arms (nodes). Since most CMAB works study a linear function reward with semi-bandit feedback, we consider two variant updates here. The first is to divide the received reward by the size of the chosen super arm and assign the quotient to each of base arms in this super arm, denoted as \textit{average}. The second is to randomly select a base arm from the super arm and assign the received reward to this base arm only (other base arms receive a reward $0$), denoted as \textit{random}. These two variants both use linear function approximation and adopt different assignments of the received reward.

\begin{figure*}[tbh!] 
\includegraphics[width=0.245\linewidth]{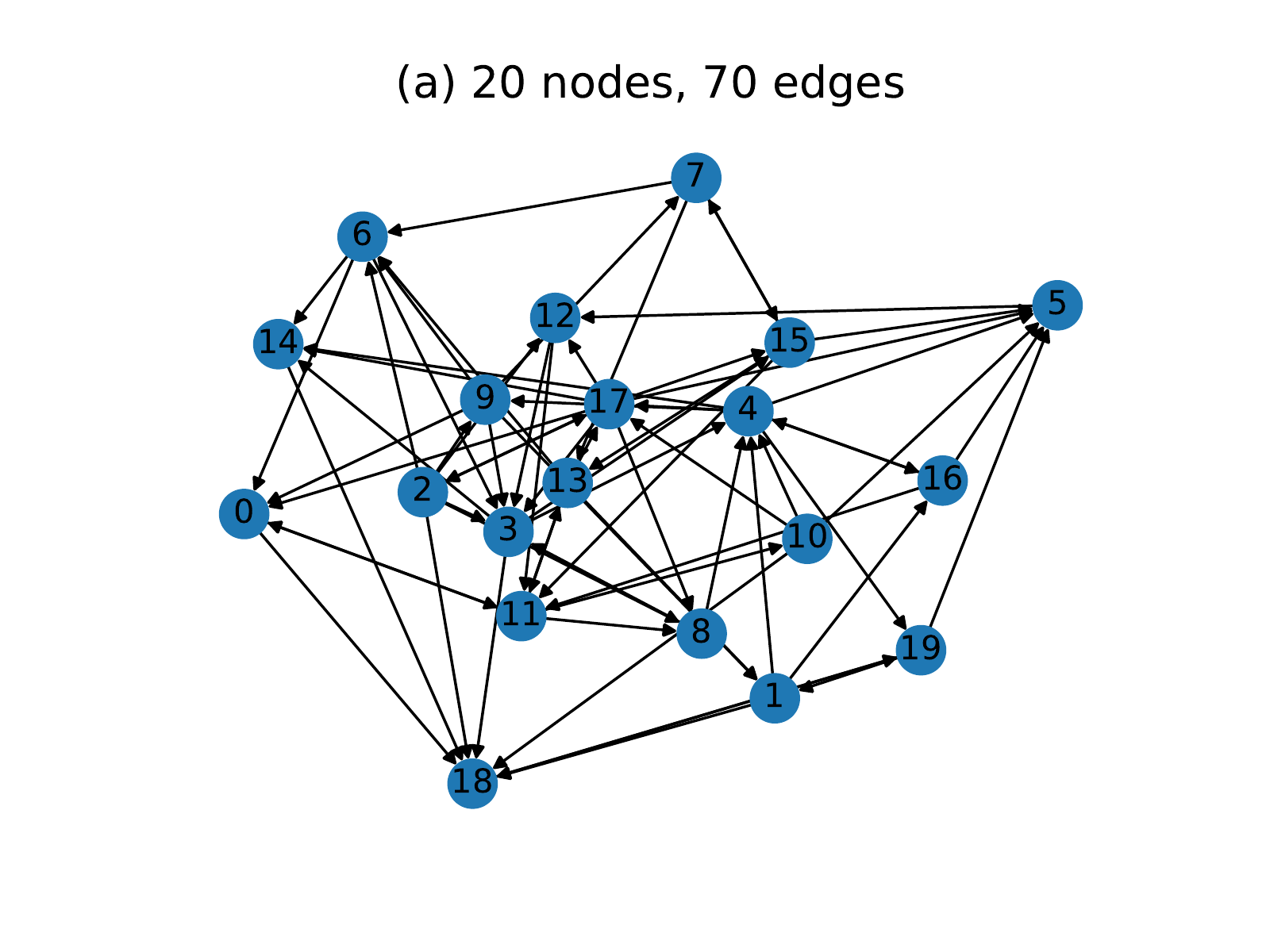} 
\includegraphics[width=0.245\linewidth]{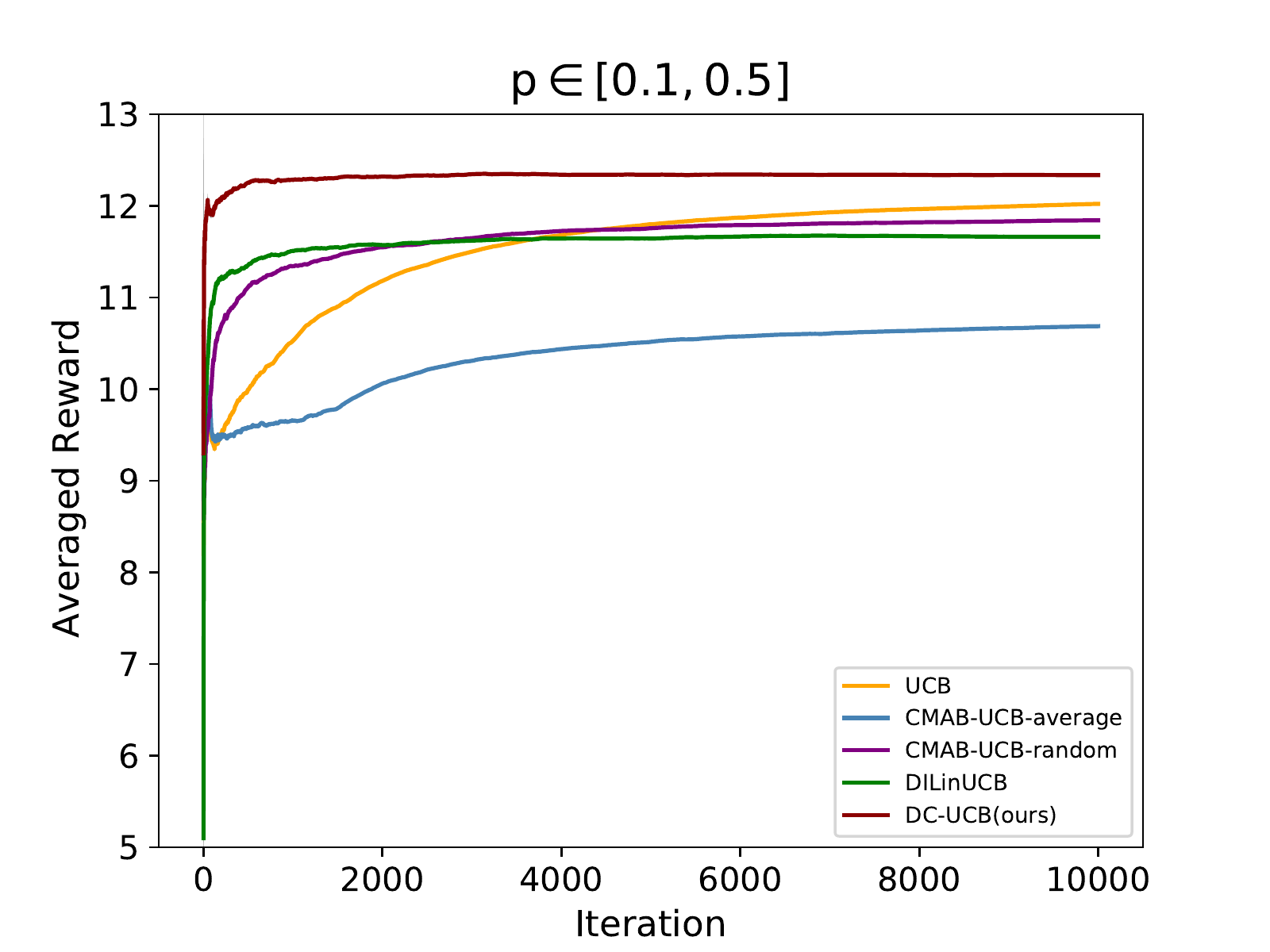}
\includegraphics[width=0.245\linewidth]{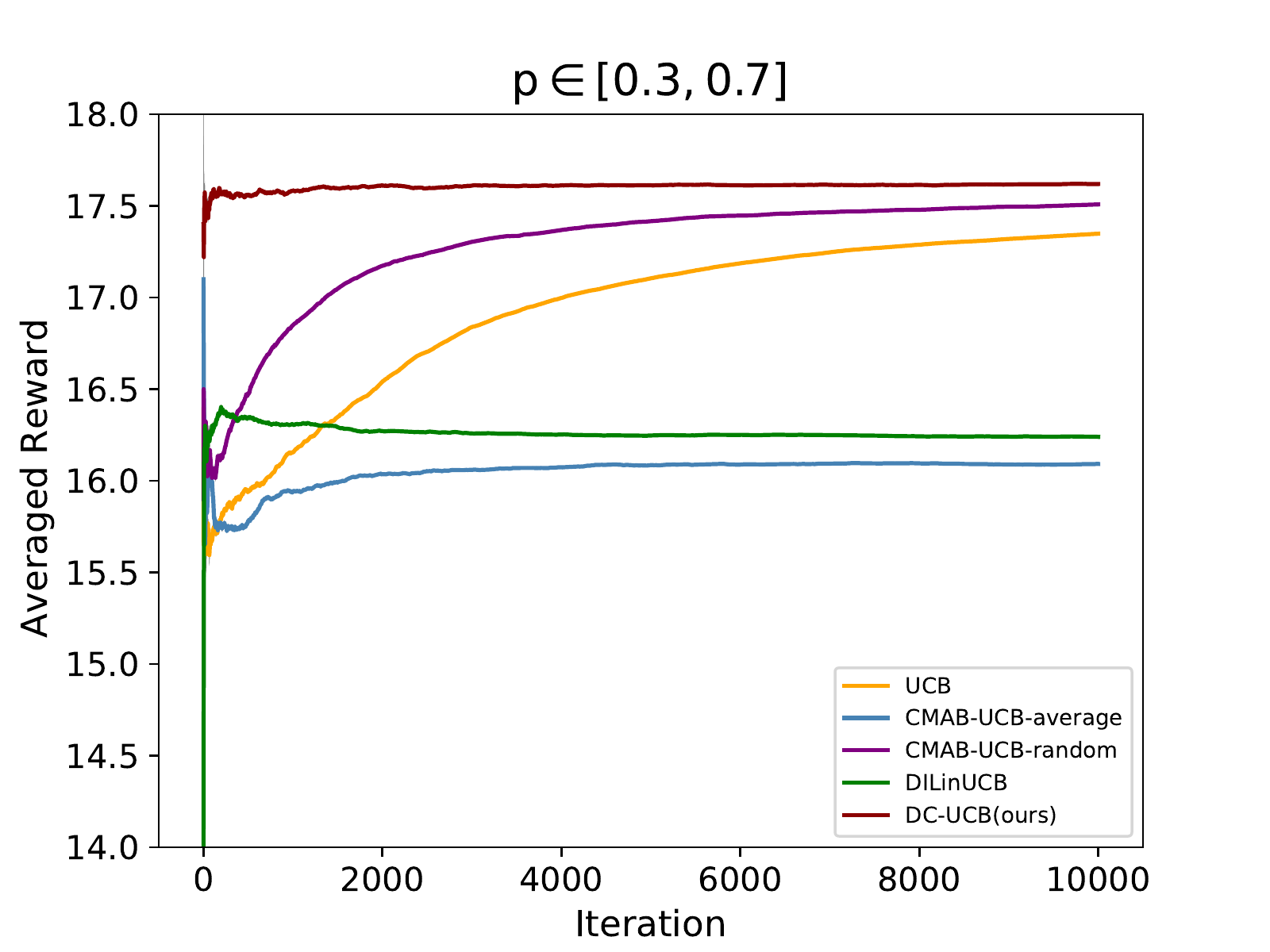} 
\includegraphics[width=0.245\linewidth]{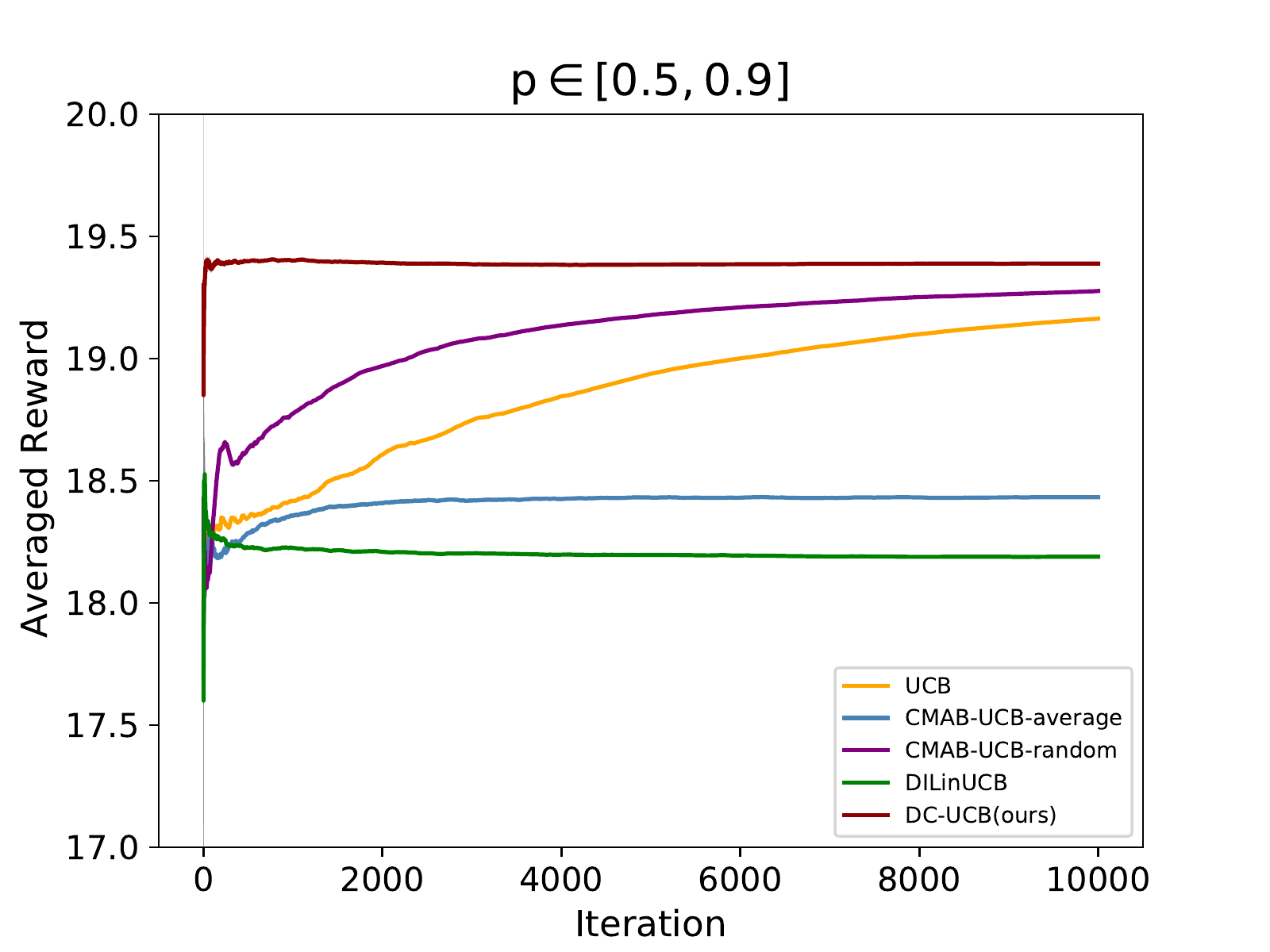}
  \caption{Comparison of DC-UCB with UCB, CMAB-UCB-average, CMAB-UCB-random and DILinUCB on synthetic networks with different activation probabilities. The performance is evaluated in averaged reward (the cumulative reward divides by the number of iterations). The standard error is represented in gray shades and all results are averaged over $10$ independent runs. }
  \label{fig:synthetic}
\end{figure*}

 % \textit{DILinUCB.~} 
\indent \textit{DILinUCB.~} This algorithm \citep{Model-Independent2017} can also be applied to our DC model, though its heuristic objective has no guaranteed approximation ratio. In this setting, the feedback for each pair of the seed node and a node is $1$ if there is an influence path from the seed node to it under our (partial) edge-level feedback. We adopt the greedy oracle designed in \cite{Model-Independent2017} as an offline oracle. For the tabular case, we simply use the one-hot representations as features.

% \textit{CUCB(IC), IMLinUCB(IC) and IMFB(IC).~}
\indent \textit{CUCB(IC), IMLinUCB(IC), and IMFB(IC).~} Recall that if the probabilities of each node are the same, the diffusion process under DC and IC are equivalent. In this special case, we can compare with the state-of-the-art IC-based algorithms CUCB \citep{WeiChen2017}, IMLinUCB \citep{zhengwen2017nips} and IMFB \citep{IMFB2019}. 
For the offline oracle required by these algorithms, we adopt the greedy algorithm \citep{Kempe2003} as is in our DC-UCB.

\subsection{Synthetic Network}
In this experiment, we compare the performance of our DC-UCB algorithm with other baselines on the synthetic network with different activation probabilities.

For the synthetic network, we randomly generate an Erd\"os-R\'enyi graph with $p=0.2$ for $n=20$ nodes. The resulting network contains $m=70$ edges, as shown in Figure \ref{fig:synthetic}(a).
For this network, we set up three groups of activation probabilities falling into different intervals.
Specifically, each activation probability is uniformly sampled from $[0.1,0.5]$,$[0.3,0.7]$,$[0.5,0.9]$ in three groups, respectively. 
Different values of activation probabilities correspond to different ability levels of the social network to spread information. 
Thus the performance of algorithms in these settings could represent their learning ability in different environments. 
To maintain the decreasing property of the DC model, the probability sequence of each node is then sorted in decreasing order. The seed set cardinality is set to $K=2$ under all three settings.

We compare the performance of our DC-UCB with UCB, CMAB-UCB-average, CMAB-UCB-random and DILinUCB, when solving the OIM problem under DC. All algorithms run for $T=10k$ rounds. 
The averaged rewards of those algorithms defined as the cumulative reward divides by the number of iterations are shown in Figure \ref{fig:synthetic}(b)(c)(d). All results are averaged over $10$ independent runs. 

Benefitting from the careful consideration of the decreasing property in DC, our DC-UCB achieves the best performance over all baselines in three settings.
The baselines UCB and CMAB-UCB-random also have comparable performance finally, but they converge much slower. This is because these two algorithms need to evaluate the reward of each seed set/node, which costs a lot of rounds to get accurate estimations. 
Especially the baseline UCB is not realistic to apply in larger networks due to the exponential number of seed sets to be evaluated. 
Other two baselines are at least $5.77\%$ ($7.83\%$, $5.21\%$) lower than ours in group $[0.1,0.5]$ ($[0.3,0.7]$, $[0.5,0.9]$, respectively).

\begin{table*}[]
\centering
\renewcommand{\arraystretch}{1.08}
\begin{tabular}{m{4cm}<{\centering} | m{1.5cm}<{\centering}  m{1.5cm}<{\centering}  m{1.5cm}<{\centering}  m{1.5cm}<{\centering}  m{1.5cm}<{\centering}  m{1.5cm}<{\centering} }
\toprule
                 & \multicolumn{1}{c}{\begin{tabular}[c]{@{}c@{}}NetHEPT\\ {[}0.1,0.5{]}\end{tabular}} & \multicolumn{1}{c}{\begin{tabular}[c]{@{}c@{}}NetHEPT\\ {[}0.3,0.7{]}\end{tabular}} & \multicolumn{1}{c}{\begin{tabular}[c]{@{}c@{}}NetHEPT\\ {[}0.5,0.9{]}\end{tabular}} & \multicolumn{1}{c}{\begin{tabular}[c]{@{}c@{}}Flickr\\ {[}0.1,0.5{]}\end{tabular}} & \multicolumn{1}{c}{\begin{tabular}[c]{@{}c@{}}Flickr\\ {[}0.3,0.7{]}\end{tabular}} & \multicolumn{1}{c}{\begin{tabular}[c]{@{}c@{}}Flickr\\ {[}0.5,0.9{]}\end{tabular}} \\
                 \hline
DC-UCB(ours)     & \textbf{209.20} &  \textbf{228.44}  & \textbf{258.01} & \textbf{278.16}  & \textbf{295.48}  & \textbf{310.44}  \\
DILinUCB         &  179.17 & 208.27  & 229.78  & 263.45   & 283.36  & 298.60  \\
CMAB-UCB-average &  119.73 & 146.01 & 160.35  &  200.19   &  220.18 & 230.93 \\
CMAB-UCB-random  &  174.47 & 209.62 & 227.17  &  252.03   &  274.75 & 284.76 \\
\bottomrule                                                                             
\end{tabular}
\caption{Comparison of DC-UCB with DILinUCB, CMAB-UCB-average and CMAB-UCB-random on NetHEPT and Flickr datasets with different activation probabilities. The performance is evaluated in averaged reward. All algorithms run for $10k$ rounds and all results are averaged over $10$ independent runs. The UCB algorithm is not included due to the exponential number of actions.}
\label{table:real:simulation} 
\end{table*}

\subsection{Real Networks}
\label{sec:simulation}
We then compare the performance of our DC-UCB with other related baselines on two real networks, NetHEPT\footnote{https://snap.stanford.edu/data/cit-HepTh.html} and Flickr\footnote{https://snap.stanford.edu/data/web-flickr.html}. 

The original NetHEPT (Flickr) dataset contains $27,770$ ($105,938$) nodes and $352,807$ ($2,316,948$, respectively) edges. 
Since it will be easier for the learning algorithm to identify the optimal seed set in the original sparse graph, here we extract a relatively dense one to make the learning task more challenging.
We first select nodes whose degree (\textit{in-degree} plus \textit{out-degree}) is in $[20,120]$, then randomly select $10$ nodes among them and keep all edges that have a start node or end node in these $10$ nodes as an intermediate graph. 
The largest connected subgraph of it forms our final network.
The resulting subgraph of NetHEPT is composed of $n=323$ nodes and $m=3,478$ edges, the subgraph of Flickr is composed of $n=319$ nodes and $m=5,904$ edges. 

We again set up three different groups with activation probabilities falling into intervals $[0.1,0.5]$,$[0.3,0.7]$,$[0.5,0.9]$. The probability sequence for each node is then sorted in decreasing order to maintain the decreasing property of DC. The seed set cardinality is set to $K=5$ in both subgraphs of NetHEPT and Flickr under three different settings. 

We compare the performance of our DC-UCB with DILinUCB, CMAB-UCB-average and CMAB-UCB-random. 
Since UCB needs to enumerate the exponential number of seed sets, here we do not include this baseline. 
All algorithms run for $T=10k$ rounds. The averaged rewards of those algorithms are shown in Table \ref{table:real:simulation} and all results are averaged over $10$ independent runs.

Our DC-UCB algorithm again shows consistent advantages over other baselines in all six settings, which demonstrates its strong learning ability in different environments. The CMAB-UCB-random performs the second-best in the setting NetHEPT$[0.3, 0.7]$, but is $8.24\%$ lower than ours. In the other settings, DILinUCB performs the second-best but is at least $3.81\%$ lower than DC-UCB. The baseline CMAB-UCB-average performs worse and more than $25.48\%$ lower than ours in all six settings.

\begin{table*}  
\centering
\renewcommand{\arraystretch}{1.1}
\begin{tabularx}{17cm}{l|cccccc}   
\toprule                      
    & NetHEPT($0.2$)  & NetHEPT($0.5$) & NetHEPT($0.8$)  & Flickr($0.2$)  & Flickr($0.5$) & Flickr($0.8$) \\
\hline 
DC-UCB(ours) & \textbf{137.43}  & \textbf{230.96}  & \textbf{256.99}  & \textbf{234.72}  & \textbf{292.02}  & \textbf{309.32}  \\
CUCB(IC) & 132.01  &  227.40  & 252.76  & 223.31 & 283.82 & 304.20 \\
IMLinUCB(IC) & 130.64   & 224.25  & 251.87  & 231.36 & 284.55 & 305.18 \\
IMFB (IC) & 89.58   & 211.47  & 235.60  & 178.96 & 238.24 & 261.57 \\
DILinUCB &  131.75  & 204.78  & 232.31  & 220.98 & 277.02 & 288.89 \\
CMAB-UCB-average &  68.45  & 131.51  & 160.16  & 157.48 & 219.34 & 251.20 \\
CMAB-UCB-random &  95.71   & 192.71  & 228.35  & 207.54 & 268.89 & 289.17 \\
\bottomrule 
\end{tabularx}  
\caption{Comparisons of DC-UCB with CUCB(IC), IMLinUCB(IC), IMFB(IC), DILinUCB, CMAB-UCB-average and CMAB-UCB-random under homogeneous activation probabilities. Conducted on the NetHEPT and Flickr datasets with three values of $p$ tested on each dataset. The performance is evaluated in averaged reward and all results are averaged over $10$ independent runs.}  
\label{table:sameP}
\end{table*} 

\subsection{Homogeneous Activation Probabilities} 
\label{sec:exp:homo}

Recall that when $p_v(i) \equiv p_v'$ for any $v\in V$ and $i \in [|N(v)|]$, the influence propagations under DC and IC are equivalent. Thus in this case, we can compare our DC-UCB with the CUCB(IC), IMLinUCB(IC) and IMFB(IC) directly. 

In this experiment, we adopt the same subgraphs of real networks as Section \ref{sec:simulation} but with different activation probabilities. 
We consider the case that all activation probabilities are the same, where the information diffusions are equivalent under IC and DC. Three choices of $p \equiv 0.2,0.5,0.8$ are tested on each network and the seed set cardinality is set to $K=5$ for all choices. 
Those three values also reflect different abilities of the social network to spread information and thus can well demonstrate the learning abilities of algorithms.

We compare our DC-UCB with CUCB(IC), IMLinUCB(IC), IMFB(IC), DILinUCB, CMAB-UCB-average and CMAB-UCB-random.  
The dimension is set to $d=5$ in IMFB(IC). As for IMLinUCB(IC), since its tabular case is equivalent to CUCB(IC) and the huge number of edges could result in high computational complexity, we randomly generate a $5$-dimensional feature vector for each edge as input to improve its learning efficiency. 
All algorithms run for $T=10k$ rounds. The averaged rewards of those algorithms are shown in Table \ref{table:sameP}, which are averaged over $10$ independent runs. Again UCB is not tested due to the exponential number of seed sets.

Our DC-UCB algorithm performs better than all baselines in six environments. These results indicate that in the setting where IC and DC are equivalent, our DC-based algorithm is more efficient than that of IC-based algorithms. 
The reason is that under DC, the leading probabilities of each node would receive more updates, since those probabilities would be always observed once the node has active in-neighbors. 
And as the leading probabilities are much more important than the tail probabilities, the influence spread would be estimated more accurately. 
While under IC, each activation probability term is bound to a specific edge, thus the updates are performed in a uniform manner over all incoming edges, making the estimated influence spread less accurate when observing the same propagations with DC.

\section{Conclusion}

This work is the first to study the OIM problem under the DC model, which generalizes IC by removing its edge independence assumption. The DC model is general enough to consider the decay property and the market saturation phenomenon of real information diffusion. 
We propose the DC-UCB algorithm to solve this problem with rigorous regret bound guarantees. Compared with the regret order of that under IC, our regret bound is at least as good. 
The algorithm is tested extensively on both synthetic datasets and real datasets of NetHEPT and Flickr against several baselines. Our algorithm consistently outperforms the baselines by a significant margin which validates its practical effectiveness.

An interesting future direction is to consider the influence factor of the start node $u$ when $u$ tries to activate $v$. 
This remains an instance of DC but is more general than the assumption in our work. The generalization is very important as the activation probabilities between users are likely to depend on the influence ability of the start node, apart from the end node, in real diffusion problems.
It will be more challenging to derive a general formulation incorporating such a factor with the decreasing and order-independence requirement of the DC model.

\bibliographystyle{named}
\bibliography{ref}

\begin{thebibliography}{}

\bibitem[\protect\citeauthoryear{Auer \bgroup \em et al.\egroup
  }{2002}]{auer2002finite}
Peter Auer, Nicolo Cesa-Bianchi, and Paul Fischer.
\newblock Finite-time analysis of the multiarmed bandit problem.
\newblock {\em Machine learning}, 47(2-3):235--256, 2002.

\bibitem[\protect\citeauthoryear{Cautis \bgroup \em et al.\egroup
  }{2019}]{cautis2019adaptive}
Bogdan Cautis, Silviu Maniu, and Nikolaos Tziortziotis.
\newblock Adaptive influence maximization.
\newblock In {\em Proceedings of the 25th ACM SIGKDD International Conference
  on Knowledge Discovery \& Data Mining}, pages 3185--3186, 2019.

\bibitem[\protect\citeauthoryear{Chen \bgroup \em et al.\egroup
  }{2010}]{chen2010scalableLT}
Wei Chen, Yifei Yuan, and Li~Zhang.
\newblock Scalable influence maximization in social networks under the linear
  threshold model.
\newblock In {\em Proceedings of the 10th International Conference on Data
  Mining}, pages 88--97, 2010.

\bibitem[\protect\citeauthoryear{Chen \bgroup \em et al.\egroup
  }{2013a}]{chen2013information}
Wei Chen, Carlos Castillo, and Laks V.~S. Lakshmanan.
\newblock {\em Information and influence propagation in social networks}.
\newblock Morgan \& Claypool Publishers, 2013.

\bibitem[\protect\citeauthoryear{Chen \bgroup \em et al.\egroup
  }{2013b}]{WeiChen2013}
Wei Chen, Yajun Wang, and Yang Yuan.
\newblock Combinatorial multi-armed bandit: General framework, results and
  applications.
\newblock In {\em Proceedings of the 30th International Conference on Machine
  Learning}, pages 151--159, 2013.

\bibitem[\protect\citeauthoryear{Chen \bgroup \em et al.\egroup
  }{2016}]{WeiChen2016}
Wei Chen, Yajun Wang, Yang Yuan, and Qinshi Wang.
\newblock Combinatorial multi-armed bandit and its extension to
  probabilistically triggered arms.
\newblock {\em The Journal of Machine Learning Research}, 17(1):1746--1778,
  2016.

\bibitem[\protect\citeauthoryear{Goyal \bgroup \em et al.\egroup
  }{2010}]{goyal2010learning}
Amit Goyal, Francesco Bonchi, and Laks V.~S. Lakshmanan.
\newblock Learning influence probabilities in social networks.
\newblock In {\em Proceedings of the 3rd ACM international conference on Web
  search and data mining}, pages 241--250, 2010.

\bibitem[\protect\citeauthoryear{Goyal \bgroup \em et al.\egroup
  }{2011}]{goyal2011simpathLT}
Amit Goyal, Wei Lu, and Laks V.~S. Lakshmanan.
\newblock Simpath: An efficient algorithm for influence maximization under the
  linear threshold model.
\newblock In {\em Proceedings of the 11th International Conference on Data
  Mining}, pages 211--220, 2011.

\bibitem[\protect\citeauthoryear{Gruhl \bgroup \em et al.\egroup
  }{2004}]{gruhl2004information}
Daniel Gruhl, Ramanathan Guha, David Liben-Nowell, and Andrew Tomkins.
\newblock Information diffusion through blogspace.
\newblock In {\em Proceedings of the 13th international conference on World
  Wide Web}, pages 491--501, 2004.

\bibitem[\protect\citeauthoryear{Han \bgroup \em et al.\egroup
  }{2018}]{han2018efficient}
Kai Han, Keke Huang, Xiaokui Xiao, Jing Tang, Aixin Sun, and Xueyan Tang.
\newblock Efficient algorithms for adaptive influence maximization.
\newblock {\em Proceedings of the VLDB Endowment}, 11(9):1029--1040, 2018.

\bibitem[\protect\citeauthoryear{Hoeffding}{1963}]{hoeffding1963probability}
Wassily Hoeffding.
\newblock Probability inequalities for sums of bounded random variables.
\newblock {\em Journal of the American Statistical Association},
  58(301):13--30, 1963.

\bibitem[\protect\citeauthoryear{Hogg and Lerman}{2009}]{hogg2009stochastic}
Tad Hogg and Kristina Lerman.
\newblock Stochastic models of user-contributory web sites.
\newblock {\em arXiv preprint arXiv:0904.0016}, 2009.

\bibitem[\protect\citeauthoryear{Jung \bgroup \em et al.\egroup
  }{2012}]{jung2012irie}
Kyomin Jung, Wooram Heo, and Wei Chen.
\newblock Irie: Scalable and robust influence maximization in social networks.
\newblock In {\em Proceedings of the 12th International Conference on Data
  Mining}, pages 918--923. IEEE, 2012.

\bibitem[\protect\citeauthoryear{Kempe \bgroup \em et al.\egroup
  }{2003}]{Kempe2003}
David Kempe, Jon Kleinberg, and {\'E}va Tardos.
\newblock Maximizing the spread of influence through a social network.
\newblock In {\em Proceedings of the 9th ACM SIGKDD International Conference on
  Knowledge Discovery \& Data Mining}, pages 137--146, 2003.

\bibitem[\protect\citeauthoryear{Kempe \bgroup \em et al.\egroup
  }{2005}]{kempe2005influential}
David Kempe, Jon Kleinberg, and {\'E}va Tardos.
\newblock Influential nodes in a diffusion model for social networks.
\newblock In {\em International Colloquium on Automata, Languages, and
  Programming}, pages 1127--1138. Springer, 2005.

\bibitem[\protect\citeauthoryear{Kossinets and
  Watts}{2006}]{kossinets2006empirical}
Gueorgi Kossinets and Duncan~J Watts.
\newblock Empirical analysis of an evolving social network.
\newblock {\em Science}, 311(5757):88--90, 2006.

\bibitem[\protect\citeauthoryear{Lei \bgroup \em et al.\egroup
  }{2015}]{lei2015online}
Siyu Lei, Silviu Maniu, Luyi Mo, Reynold Cheng, and Pierre Senellart.
\newblock Online influence maximization.
\newblock In {\em Proceedings of the 21th ACM SIGKDD International Conference
  on Knowledge Discovery \& Data Mining}, pages 645--654, 2015.

\bibitem[\protect\citeauthoryear{Leskovec \bgroup \em et al.\egroup
  }{2007}]{leskovec2007dynamics}
Jure Leskovec, Lada~A. Adamic, and Bernardo~A. Huberman.
\newblock The dynamics of viral marketing.
\newblock {\em ACM Transactions on the Web (TWEB)}, 1(1):5--es, 2007.

\bibitem[\protect\citeauthoryear{Li \bgroup \em et al.\egroup
  }{2020}]{li2020online}
Shuai Li, Fang Kong, Kejie Tang, Qizhi Li, and Wei Chen.
\newblock Online influence maximization under linear threshold model.
\newblock In {\em Advances in Neural Information Processing Systems}, 2020.

\bibitem[\protect\citeauthoryear{Myers \bgroup \em et al.\egroup
  }{2012}]{myers2012information}
Seth~A. Myers, Chenguang Zhu, and Jure Leskovec.
\newblock Information diffusion and external influence in networks.
\newblock In {\em Proceedings of the 18th ACM SIGKDD International Conference
  on Knowledge Discovery \& Data Mining}, pages 33--41, 2012.

\bibitem[\protect\citeauthoryear{Netrapalli and
  Sanghavi}{2012}]{netrapalli2012learning}
Praneeth Netrapalli and Sujay Sanghavi.
\newblock Learning the graph of epidemic cascades.
\newblock {\em ACM SIGMETRICS Performance Evaluation Review}, 40(1):211--222,
  2012.

\bibitem[\protect\citeauthoryear{Steeg \bgroup \em et al.\egroup
  }{2011}]{steeg2011stops}
Greg~Ver Steeg, Rumi Ghosh, and Kristina Lerman.
\newblock What stops social epidemics?
\newblock In {\em Proceedings of the 5th International AAAI Conference on
  Weblogs and Social Media}, 2011.

\bibitem[\protect\citeauthoryear{Tang \bgroup \em et al.\egroup
  }{2015}]{IMM2015}
Youze Tang, Yanchen Shi, and Xiaokui Xiao.
\newblock Influence maximization in near-linear time: A martingale approach.
\newblock In {\em Proceedings of the 2015 ACM SIGMOD International Conference
  on Management of Data}, pages 1539--1554, 2015.

\bibitem[\protect\citeauthoryear{Vaswani \bgroup \em et al.\egroup
  }{2015}]{Sharan2015-nodelevel}
Sharan Vaswani, Laks V.~S. Lakshmanan, Mark Schmidt, et~al.
\newblock Influence maximization with bandits.
\newblock {\em arXiv preprint arXiv:1503.00024}, 2015.

\bibitem[\protect\citeauthoryear{Vaswani \bgroup \em et al.\egroup
  }{2017}]{Model-Independent2017}
Sharan Vaswani, Branislav Kveton, Zheng Wen, Mohammad Ghavamzadeh, Laks V.~S.
  Lakshmanan, and Mark Schmidt.
\newblock Model-independent online learning for influence maximization.
\newblock In {\em Proceedings of the 34th International Conference on Machine
  Learning}, pages 3530--3539. JMLR. org, 2017.

\bibitem[\protect\citeauthoryear{Wang and Chen}{2017}]{WeiChen2017}
Qinshi Wang and Wei Chen.
\newblock Improving regret bounds for combinatorial semi-bandits with
  probabilistically triggered arms and its applications.
\newblock In {\em Advances in Neural Information Processing Systems}, pages
  1161--1171, 2017.

\bibitem[\protect\citeauthoryear{Wang \bgroup \em et al.\egroup
  }{2012}]{wang2012scalableIC}
Chi Wang, Wei Chen, and Yajun Wang.
\newblock Scalable influence maximization for independent cascade model in
  large-scale social networks.
\newblock {\em Data Mining and Knowledge Discovery}, 25(3):545--576, 2012.

\bibitem[\protect\citeauthoryear{Watts and Dodds}{2007}]{watts2007influentials}
Duncan~J Watts and Peter~Sheridan Dodds.
\newblock Influentials, networks, and public opinion formation.
\newblock {\em Journal of consumer research}, 34(4):441--458, 2007.

\bibitem[\protect\citeauthoryear{Wen \bgroup \em et al.\egroup
  }{2017}]{zhengwen2017nips}
Zheng Wen, Branislav Kveton, Michal Valko, and Sharan Vaswani.
\newblock Online influence maximization under independent cascade model with
  semi-bandit feedback.
\newblock In {\em Advances in neural information processing systems}, pages
  3022--3032, 2017.

\bibitem[\protect\citeauthoryear{Wortman}{2008}]{wortman2008viral}
Jennifer Wortman.
\newblock Viral marketing and the diffusion of trends on social networks.
\newblock {\em Technical Reports (CIS)}, page 880, 2008.

\bibitem[\protect\citeauthoryear{Wu and Huberman}{2007}]{wu2007novelty}
Fang Wu and Bernardo~A Huberman.
\newblock Novelty and collective attention.
\newblock {\em Proceedings of the National Academy of Sciences},
  104(45):17599--17601, 2007.

\bibitem[\protect\citeauthoryear{Wu \bgroup \em et al.\egroup
  }{2019}]{IMFB2019}
Qingyun Wu, Zhige Li, Huazheng Wang, Wei Chen, and Hongning Wang.
\newblock Factorization bandits for online influence maximization.
\newblock In {\em Proceedings of the 25th ACM SIGKDD International Conference
  on Knowledge Discovery \& Data Mining}, pages 636--646, 2019.

\bibitem[\protect\citeauthoryear{Zhang \bgroup \em et al.\egroup
  }{2016}]{zhang2016influence}
Zhijian Zhang, Hong Wu, Kun Yue, Jin Li, and Weiyi Liu.
\newblock Influence maximization for cascade model with diffusion decay in
  social networks.
\newblock In {\em International Conference of Pioneering Computer Scientists,
  Engineers and Educators}, pages 418--427, 2016.

\bibitem[\protect\citeauthoryear{Zhang \bgroup \em et al.\egroup
  }{2022}]{zhang2022online}
Zhijie Zhang, Wei Chen, Xiaoming Sun, and Jialin Zhang.
\newblock Online influence maximization with node-level feedback using standard
  offline oracles.
\newblock In {\em Proceedings of the AAAI Conference on Artificial
  Intelligence}, volume~36, pages 9153--9161, 2022.

\end{thebibliography}

\appendix

\section{Important Lemmas}

We first introduce some important lemmas which are useful in the main proof. 
The following lemma provides an important property of the influence spread - monotonicity, which plays a significant role in the analysis. 

\begin{lemma}[Monotonicity of the influence spread under the DC model]\label{lem:mono}
The influence spread $r(S,p)$ is monotonically increasing in the activation probabilities $p$. That is, for any $p$ and $p'$ satisfying $p_v(i) \le p'_v(i)$, $\forall v \in V,i \in [|N(v)|]$, $r(S,p) \le r(S,p')$ holds. 
\end{lemma}
\begin{proof}
Define $r(S,p,v)$ as the probability that node $v$ is finally influenced in the diffusion process starting from $S$ under $p$. Then the influence spread can be decomposed as $r(S,p) = \sum_{v \in V}r(S,p,v)$. Consider two diffusion process under $p$ and $p'$, it is straightforward that the probability that a node is influenced under $p'$ is larger than the probability under $p$. Thus $r(S,p) \le r(S,p')$ holds.
\end{proof}

We then arm ourselves with a technical lemma for DC-UCB. 

\begin{lemma}\label{lem:p:p'}
In both DC-UCB and DC-LinUCB, for any node $v \in V$, if $p_v(i)\le\bar{p}_v(i)$ holds for any $i\in [|N(v)|]$, then $p_v(i)\le\bar{p}'_v(i)$ holds for any $i\in [|N(v)|]$. 
\end{lemma}
\begin{proof}
Recall that $\bar{p}'_v = (\bar{p}'_v(i))_{i \in [|N(v)|]}$ in both DC-UCB and DC-LinUCB is computed as:
\begin{align*}
  &\bar{p}'_v(1) = \bar{p}_v(1) \,,\\
  &\bar{p}'_v(i) =  \min \set{\bar{p}'_v(i-1), \bar{p}_v(i)  }, \forall i=2,3,\ldots,|N(v)| \,.
\end{align*}

For $i=1$, it is obvious that 
\begin{align}
  p_v(1) \le \bar{p}_v(1) = \bar{p}'_v(1) \label{eq:pv1} \,.
\end{align}

For $i=2$, due to the decreasing property of the activation probability, it holds that $p_v(2) \le p_v(1)$. Combining \eqref{eq:pv1} we get $p_v(2) \le \bar{p}'_v(1)$. Thus, we conclude that $p_v(2) \le \min\set{\bar{p}'_v(1), \bar{p}_v(2)}  = \bar{p}'_v(2)$.

By induction, repeating the above process for $i=3,4,\ldots,|N(v)|$, we get the desired result. 
\end{proof}

\section{Proof of Theorem \ref{thm:UCB:free} and \ref{thm:UCB:dependent}}

Recall that $P_{v,i}^S$ is the probability that node $v$'s $i$-th attempt can be observed in the diffusion process starting from $S$ under activation probability vector $p$. 
According to each entry $(v,i)$, where $v \in V,i \in [|N(v)|]$, we divide the action set $\cS$ into several groups. 

\begin{definition}[Action Groups]\label{def:group}
	For any $v \in V,i \in [|N(v)|]$ and a positive number $j$, define the group of actions
	\begin{align}
		\cS_{v,i}^j = \set{S \in \cS: 2^{-j}\le P_{v,i}^S < 2^{-j+1}}\,.
	\end{align}
	By the definition, $\set{\cS_{v,i}^j}_{j\ge 1}$ forms a partition of $\set{S\in \cS:P_{v,i}^S>0}$. 
\end{definition}

For each group $\cS_{v,i}^j$, we maintain the counter $N_{v,i}^j$ to denote the number of selections of seed sets in this group and use $N_{v,i,t}^j$ to represent the value of $N_{v,i}^j$ in round $t$. Then for any $v \in V,i \in [|N(v)|], j\ge 1$, $N_{v,i,t}^j$ can be computed as follows.
\begin{align}
	N_{v,i,t}^j=
	\begin{cases}
		0,                &t=0\,,  \\
		N_{v,i,t-1}^j+1 , & t>0 \text{ and }S_t \in \cS_{v,i}^j\,, \\
		N_{v,i,t-1}^j,    &\text{ otherwise }.
	\end{cases}
\end{align} 

Before we prove of the main results, we first clarify that $\hat{p}_{t,v}(i),\bar{p}_{t,v}(i), \bar{p}'_{t,v}(i), \rho_{t,v}(i)$ represent the value of $\hat{p}_{v}(i),\bar{p}_{v}(i), \bar{p}'_{v}(i), \rho_{v}(i)$, respectively, in round $t$ for any $v \in V,i \in [|N(v)|]$. And $Y_{t,v}(i)$ represents the observed status $Y_v(i)$ for any observed $i$-th attempt of node $v$ in round $t$. 

Define the event of failure as 
\begin{align}
	\cF_{1,t} = \set{\exists v \in V, i \in [|N(v)|]: \abs{ \hat{p}_{t,v}(i)-p_v(i)}>\rho_{t,v}(i) } \,.
\end{align}
According to \cite[Lemma 3]{WeiChen2017}, it holds that for any round $t \ge 1$: 
\begin{align}
	\PP{\cF_{1,t}} \le \frac{2m}{t^2}\,.  \label{eq:thm1:F1}
\end{align}
Also, given a series $\set{j_{v,i}^{\max}}_{v \in V,i \in [|N(v)|]}$, define the event of failure about relationships between the action selection and the local entry as 
\begin{align}
	\cF_{2,t} = \set{\forall 1\le  j \le j_{v,i}^{\max}: \sqrt{\frac{6 \ln t}{\frac{1}{3}N_{v,i,t-1}^j \cdot 2^{-j}  }} > 1 \rightarrow \rho_{t,v}(i) \le \sqrt{\frac{3 \ln t}{\frac{2}{3}N_{v,i,t-1}^j \cdot 2^{-j}  }}     } \,.
\end{align}
According to \cite[Lemma 4]{WeiChen2017}, it holds that for given sequence $\set{j_{v,i}^{\max}}_{v \in V,i \in [|N(v)|]}$, 
\begin{align}
	\PP{\cF_{2,t}} \le \sum_{v \in V,i \in [|N(v)|]}\frac{j_{v,i}^{\max}} {t^2} \,. \label{eq:thm1:F2}
\end{align}

% For the adopted oracle, define $\cF_t = \set{r(S_t,\bar{p}'_t) < \alpha\cdot \opt_{\bar{p}'} }$ as the failure event that the $\oracle$ fails to output an $\alpha$-approximation solution. According to the definition of the $(\alpha,\beta)$-approximation oracle, $\PP{\cF_t^c} \le 1-\beta$. 

Now we are ready to prove the main result of DC-UCB. Define the event $\cF_t = \set{ r(S_t,\bar{p}'_t) < \alpha \cdot \opt_{\bar{p}'_t} }$. Then according to the definition of the $(\alpha,\beta)$-approximation oracle, $\PP{\cF_t} < 1-\beta$ for any $t$.

Since
\begin{align}
	\EE{r(S_t, \bar{p}'_t)} = \EE{r(S_t, \bar{p}'_t)\mid \cF_t} \PP{\cF_t} + \EE{r(S_t, \bar{p}'_t)\mid \cF_t^c} \PP{\cF_t^c} \ge \beta \cdot \EE{r(S_t, \bar{p}'_t)\mid \cF_t^c}\,, \notag
\end{align}
the $(\alpha,\beta)$-scaled regret of DC-UCB can be bounded by
\begin{align}
	R(T) &= \EE{\sum_{t=1}^T \alpha\beta\cdot \opt_p - r(S_t,p) } \notag \\
	&\le \beta\cdot \EE{ \sum_{t=1}^T \alpha\cdot \opt_{p} - r(S_t,p) \mid \cF_t^c } \notag \\
	&\le \EE{ \sum_{t=1}^T \alpha\cdot \opt_{p} - r(S_t,p) \mid \cF_t^c,\cF_{1,t}^c,\cF_{2,t}^c } + \sum_{t=1}^T \PP{\cF_{2,t}}\cdot \Delta_{\max} + \sum_{t=1}^T \PP{\cF_{1,t}}\cdot \Delta_{\max} \notag \\
	&\le \EE{ \sum_{t=1}^T \alpha\cdot \opt_{p} - r(S_t,p) \mid \cF_t^c,\cF_{1,t}^c,\cF_{2,t}^c } + \frac{\pi^2}{3}m\Delta_{\max} + \frac{\pi^2}{6}\sum_{v \in V,i \in [|N(v)|]} j_{v,i}^{\max}\cdot \Delta_{\max} \label{eq:thm1:decompose2} \\
	&= \EE{ \sum_{t=1}^T \Delta_{S_t} \mid \cF_t^c,\cF_{1,t}^c,\cF_{2,t}^c } + \frac{\pi^2}{3}m\Delta_{\max} + \frac{\pi^2}{6}\sum_{v \in V,i \in [|N(v)|]} j_{v,i}^{\max}\cdot \Delta_{\max} \,,\label{eq:thm1:decompose}
\end{align}
where \eqref{eq:thm1:decompose2} sums \eqref{eq:thm1:F1} and \eqref{eq:thm1:F2} over $T$ rounds. 

First, we consider the first part  $\EE{ \sum_{t=1}^T \Delta_{S_t} \mid \cF_t^c,\cF_{1,t}^c,\cF_{2,t}^c }:= (*)$ of \eqref{eq:thm1:decompose}. 

For each pair $(v,i)$, we define a positive real number $M_{v,i}$. And for any seed set $S$, define $M_{S} = \max_{(v,i):P_{v,i}^S > 0 }M_{v,i}$. Specifically if $P_{v,i}^S = 0$ for any $v \in V,i \in [|N(v)|]$, then $M_{S} = 0$. Based on $M_{S_t}$, $(*)$ can be divided into two parts as
\begin{align}
	(*) \le \EE{ \sum_{t=1}^T \Delta_{S_t}\cdot\bOne{\Delta_{S_t} \ge M_{S_t}} \mid \cF_t^c,\cF_{1,t}^c,\cF_{2,t}^c } + \EE{ \sum_{t=1}^T \Delta_{S_t}\cdot\bOne{\Delta_{S_t} < M_{S_t}} \mid \cF_t^c,\cF_{1,t}^c,\cF_{2,t}^c } \label{eq:thm1:DeltaSt:decompose} \,.
\end{align}
Based on event $\cF_t^c,\cF_{1,t}^c,\cF_{2,t}^c$, combining Lemma \ref{lem:p:p'} and Lemma \ref{lem:TPM}, we have
\begin{align}
	\Delta_{S_t} = \alpha\cdot \opt_p - r(S_t,p) &= \alpha \cdot r(S^\opt_p,p) - r(S_t,p) \notag \\
	&\le \alpha\cdot r(S^\opt_p,\bar{p}'_t) - r(S_t,p) \notag \\
	&\le \alpha\cdot r(S^\opt_{\bar{p}'_t},\bar{p}'_t) - r(S_t,p) \notag \\
	&\le r(S_t,\bar{p}'_t) - r(S_t,p)\notag  \\
	&\le r(S_t,\bar{p}_t) - r(S_t,p) \notag \\
	&\le \sum_{v \in V\setminus S_t} \sum_{u \in V_{S_t,v}} \sum_{i \in [|N(u)|]} P_{u,i}^{S_t} \cdot[\bar{p}_u(i) - p_u(i)] \notag \,.
\end{align}
Thus, for the first part of \eqref{eq:thm1:DeltaSt:decompose}, it holds that
\begin{align}
	M_{S_t} \le \Delta_{S_t} &\le \sum_{v \in V\setminus S_t} \sum_{u \in V_{S_t,v}} \sum_{i \in [|N(u)|]} P_{u,i}^{S_t} \cdot[\bar{p}_u(i) - p_u(i)] \notag \\
	&= \sum_{u \in V} \sum_{v \in V}\bOne{u \rightarrow v} \sum_{i \in [|N(u)|]} P_{u,i}^{S_t} \cdot[\bar{p}_u(i) - p_u(i)] \notag \\
	&\le \widetilde{V}\sum_{u \in V} \sum_{i \in [|N(u)|]} P_{u,i}^{S_t} \cdot[\bar{p}_u(i) - p_u(i)]  \notag \\
	&\le -M_{S_t} + 2\widetilde{V}\sum_{u \in V} \sum_{i \in [|N(u)|]} P_{u,i}^{S_t} \cdot[\bar{p}_u(i) - p_u(i)] \notag \\
	&= 2  \widetilde{V}\sum_{u \in V}\sum_{i \in [|N(u)|]} \bracket{ P_{u,i}^{S_t} \cdot[\bar{p}_u(i) - p_u(i)] -  \frac{M_{S_t}}{2\widetilde{V}m}} \notag  \,.
\end{align}
We then desire to bound $P_{u,i}^{S_t} \cdot[\bar{p}_u(i) - p_u(i)] := (\Delta)$. Denote the group index of $S_t$ based on $(u,i)$ by $j_{u,i}$, then it is straightforward by the definition of the action group (Definition \ref{def:group}) that 
\begin{align*}
 	P_{u,i}^{S_t} \le 2\cdot 2^{-j_{u,i}} \,. 
\end{align*}
Also based on the event $\cF_{1,t}^c$ and $\cF_{2,t}^c$, we can get
\begin{align*}
	\bar{p}_u(i) - p_u(i) \le \min\set{1, \sqrt{\frac{6 \ln T}{\frac{1}{3}\cdot N_{u,i,t-1}^{j_{u,i}} \cdot 2^{-j_{u,i}} }}} \,.
\end{align*}
Combining these two, $(\Delta)$ can be bounded by
\begin{align}
	(\Delta) = P_{u,i}^{S_t} \cdot[\bar{p}_u(i) - p_u(i)] \le \min\set{ 2\cdot 2^{-j_{u,i}}, \sqrt{  \frac{ 72\cdot2^{-j_{u,i}}\ln T  }{ N_{u,i,t-1}^{j_{u,i}}  }    } } \,.
\end{align}

Choosing $j_{u,i}^{\max} = \lceil \log_2\frac{4\widetilde{V}m}{M_{u,i}} \rceil$ in event $\cF_{2,t}^c$, we move on to consider different cases of $j_{u,i}$ to bound $(\Delta)$, 
\begin{itemize}
	\item If $j_{u,i} \le j_{u,i}^{\max}$,
	\begin{itemize}
		\item If $N_{u,i,t-1}^{j_{u,i}} = 0$, then $(\Delta) \le 2\cdot 2^{-j_{u,i}} $;
		\item If $0< N_{u,i,t-1}^{j_{u,i}}\le \frac{288\widetilde{V}^2m^2 2^{-j_{u,i}} \ln T }{M_{u,i}^2}:= \ell_{j_{u,i},T}(M_{u,i})$, then $(\Delta) \le \sqrt{  \frac{ 72\cdot2^{-j_{u,i}}\ln T  }{ N_{u,i,t-1}^{j_{u,i}}  }    } $ ;
		\item If $N_{u,i,t-1}^{j_{u,i}} > \ell_{j_{u,i},T}(M_{u,i})$,  then  $(\Delta) \le \frac{M_{u,i}}{2\widetilde{V}m}$.
	\end{itemize}
	\item If $j_{u,i} > j_{u,i}^{\max}$, then $(\Delta) < 2\cdot 2^{-j_{u,i}^{\max}} < \frac{M_{u,i}}{2\widetilde{V}m}$ .
\end{itemize}
Combine all these cases and define 
\begin{align}
	k_{j,T}\bracket{M, s} = \left\{
		\begin{aligned}
			& 4\widetilde{V}\cdot 2^{-j},                &s=0 \,, \\
			& 2\widetilde{V} \sqrt{ \frac{72\cdot 2^{-j}\ln T }{s} } , &0<s \le \ell_{j,T}(M)\,,\\
			& 0,    & s > \ell_{j,T}(M)\,.
		\end{aligned}
	\right.
\end{align}
We can then conclude that 
\begin{align}
	\EE{ \Delta_{S_t}\cdot \bOne{\Delta_{S_t}\ge M_{S_t}} \mid \cF_t^c,\cF_{1,t}^c,\cF_{2,t}^c } &\le 2\widetilde{V}\sum_{u \in V\setminus S_t}\sum_{i \in [|N(u)|]} \bracket{ P_{u,i}^{S_t} \cdot[\bar{p}_u(i) - p_u(i)] -  \frac{M_{S_t}}{2\widetilde{V}m}} \notag\\
	&\le \sum_{u \in V}\sum_{i \in [|N(v)|]} k_{j_{u,i}T}\bracket{M_{u,i},N_{u,i,t-1}^{j_{u,i}}} \,.
\end{align}

Further, consider the first part of $(*)$, 
\begin{align}
	\EE{\sum_{t=1}^T \Delta_{S_t}\cdot \bOne{\Delta_{S_t}\ge M_{S_t}} \mid \cF_t^c,\cF_{1,t}^c,\cF_{2,t}^c } &\le \sum_{t=1}^T \sum_{u \in V}\sum_{i \in [|N(v)|]} k_{j_{u,i}T}\bracket{M_{u,i},N_{u,i,t-1}^{j_{u,i}}} \notag \\
	&= \sum_{u \in V}\sum_{i \in [|N(u)|]} \sum_{j=1}^{\infty} \sum_{s=0}^{N_{u,i,T-1}^{j}} k_{j,T}(M_{u,i},s) \notag  \\
	&\le \sum_{u \in V}\sum_{i \in [|N(u)|]} \sum_{j=1}^{\infty} \bracket{ \sum_{s=0}^{\ell_{j,T}(M_{u,i})} k_{j,T}(M_{u,i},s)  } \label{eq:thm1:lessthanl}  \\
	&= \sum_{u \in V}\sum_{i \in [|N(u)|]} \sum_{j=1}^{\infty} \bracket{ k_{j,T}(M_{u,i},0) + \sum_{s=1}^{\ell_{j,T}(M_{u,i})} k_{j,T}(M_{u,i},s)  } \notag  \\
	&= \sum_{u \in V}\sum_{i \in [|N(u)|]} \sum_{j=1}^{\infty} \bracket{ k_{j,T}(M_{u,i},0) + \sum_{s=1}^{\ell_{j,T}(M_{u,i})} 2\widetilde{V}\sqrt{\frac{72\cdot 2^{-j}\ln T}{s}}  }  \notag  \\
	&\le \sum_{u \in V}\sum_{i \in [|N(u)|]} \sum_{j=1}^{\infty} \bracket{ k_{j,T}(M_{u,i},0) + 4\widetilde{V}\sqrt{72\cdot 2^{-j} \ln T } \sqrt{\ell_{j,T}(M_{u,i})} } \label{eq:thm1:integral}\\
	&= \sum_{u \in V}\sum_{i \in [|N(u)|]} \sum_{j=1}^{\infty} \bracket{ 4\widetilde{V}2^{-j} + \frac{576\cdot 2^{-j} \widetilde{V}^2 m \ln T}{M_{u,i}} } \label{eq:thm1:valuel} \\
	&\le \sum_{u \in V}\sum_{i \in [|N(u)|]}\bracket{ \sum_{j=1}^{\infty}2^{-j} \cdot  \bracket{4\widetilde{V} + \frac{576 \widetilde{V}^2 m \ln T}{M_{u,i}}} } \notag \\
	&= \sum_{u \in V}\sum_{i \in [|N(u)|]} \bracket{4\widetilde{V} + \frac{576 \widetilde{V}^2 m \ln T}{M_{u,i}}} \notag \\
	&= \sum_{u \in V}\sum_{i \in [|N(u)|]} \frac{576 \widetilde{V}^2 m \ln T}{M_{u,i}} + 4\widetilde{V}m \notag \,,
\end{align}
where \eqref{eq:thm1:lessthanl} holds since $k_{j,T}(M,s)=0$ when $s >\ell_{j,T}(M)$, \eqref{eq:thm1:integral} is by the fact that for any natural number $n$, 
\begin{align}
	\sum_{s=1}^n \sqrt{\frac{1}{s}} \le \int_{s=0}^n\sqrt{\frac{1}{s}}ds = 2\sqrt{n} \,,  \notag
\end{align}
and \eqref{eq:thm1:valuel} is by the definition of $\ell_{j,T}(M_{u,i}) = \frac{288\widetilde{V}^2m^2 2^{-j} \ln T }{M_{u,i}^2}$.

For any pair $(u,i)$, let $M_{u,i} = M$ and then by definition $M_{S_t}=M$, where the value of $M$ will be chosen later. Then the second part of \eqref{eq:thm1:DeltaSt:decompose} can be bounded by $TM$. Subsequently, we have
\begin{align}
	R(T) &\le (*) + \frac{\pi^2}{3}m\Delta_{\max} + \frac{\pi^2}{6}\sum_{v \in V,i \in [|N(v)|]} j_{v,i}^{\max}\cdot \Delta_{\max}  \notag \\ 
	&\le \sum_{u \in V}\sum_{i \in [|N(u)|]} \frac{576 \widetilde{V}^2 m \ln T}{M_{u,i}} + 4\widetilde{V}m + TM + \frac{\pi^2}{3}m\Delta_{\max} + \frac{\pi^2}{6}\sum_{v \in V,i \in [|N(v)|]} j_{v,i}^{\max}\cdot \Delta_{\max} \notag \\
	&= \frac{576 \widetilde{V}^2 m^2 \ln T}{M} + 4\widetilde{V}m + TM + \frac{\pi^2}{3}m\Delta_{\max} + \frac{m\pi^2\Delta_{\max}}{6}\log_2\bracket{\frac{4\widetilde{V}m}{M}}  \label{eq:thm1:defj} \\
	&\le 48\widetilde{V}m\sqrt{T\ln T} + 4 \widetilde{V}m + \frac{m\pi^2\Delta_{\max}}{6}\bracket{2+ \frac{\log_2\sqrt{\frac{576\widetilde{V}^2 m^2 \ln T}{T}}}{4 \widetilde{V}m}} \label{eq:thm1:free:result} \\
	&= O\bracket{mn\sqrt{T\ln T}} \,, \notag
\end{align}
where \eqref{eq:thm1:defj} is derived from the definition that $j_{u,i}^{\max} = \lceil \log_2\frac{4 \widetilde{V}m}{M_{u,i}} \rceil$ and $M_{u,i}=M$ for all $u \in V,i \in [|N(u)|]$ and \eqref{eq:thm1:free:result} is by choosing $M = \sqrt{\frac{576\widetilde{V}^2 m^2 \ln T}{T}}$. Above all, we obtain the problem-independent $(\alpha,\beta)$-scaled regret bound of the DC-UCB algorithm. 

For the problem-dependent regret bound, we choose $M_{u,i} = \Delta^{u,i}_{\min}$. Then according to the definition, the event $\Delta_{S_t} \le M_{S_t}$ does not happen any more. Thus the second part of \eqref{eq:thm1:DeltaSt:decompose} is $0$. Then in this case, the $(\alpha,\beta)$-scaled regret can be bounded by
\begin{align}
	R(T) &\le (*) + \frac{\pi^2}{3}m\Delta_{\max} + \frac{\pi^2}{6}\sum_{v \in V,i \in [|N(v)|]} j_{v,i}^{\max}\cdot \Delta_{\max}  \notag \\ 
	&\le \sum_{u \in V}\sum_{i \in [|N(u)|]} \frac{576 \widetilde{V}^2 m \ln T}{M_{u,i}} + 4\widetilde{V}m  + \frac{\pi^2}{3}m\Delta_{\max} + \frac{\pi^2}{6}\sum_{v \in V,i \in [|N(v)|]} j_{v,i}^{\max}\cdot \Delta_{\max} \notag \\
	&= \sum_{v \in V}\sum_{i \in [|N(v)|]} \frac{576 \widetilde{V}^2 m \ln T}{\Delta^{v,i}_{\min}} + 4\widetilde{V}m  + \frac{\pi^2\Delta_{\max}}{6} \sum_{v \in V,i \in [|N(v)|]} \bracket{ 2+\log_2 \frac{4\widetilde{V}m}{\Delta_{\min}^{v,i}} } \notag\\ 
	&= O\bracket{\frac{m^2n^2\ln T}{\Delta_{\min}}}\,. \notag
\end{align}

\end{document}